\documentclass[11pt]{article}
\usepackage{amsmath,amssymb,amsthm}
\usepackage{pstricks}
\newtheorem{thm}{Theorem}

\newtheorem{lem}[thm]{Lemma}
\newtheorem{clm}[thm]{Claim}
\theoremstyle{definition}
\newtheorem{dfn}[thm]{Definition}

\newcommand{\classfont}{\mathrm}
\newcommand{\AC}{\classfont{AC}}
\newcommand{\load}[1]{\langle\mathrm{loader}_{#1}\rangle}
\newcommand{\dist}[1]{\langle\mathrm{distributor}_{#1}\rangle}
\newcommand{\ver}[1]{\langle\mathrm{verifier}_{#1}\rangle}

\def\vdashstar{{\vdash^*}}
\def\vdashstarsmall{{\hbox to 0em{$\vdashstar$}\phantom{\vdash}}}
\def\pprime{{\prime\prime}}

\title{Unshuffling a Square is NP-Hard
\\
\small Preliminary version --- comments appreciated}
\author{Sam Buss\thanks{Supported
in part by NSF grant DMS-1101228.}
\\
\small Department of Mathematics \\
\small University of California, San Diego\\
\small La Jolla, CA 92093-0112, USA\\
\small \tt sbuss@math.ucsd.edu
\and
Michael Soltys\thanks{Supported in part
by an NSERC Discovery Grant.
This project was carried out
while the second author was visiting UCSD in Fall 2012.}
\\
\small Department of Computing \& Software \\
\small McMaster University \\
\small Hamilton, Ontario L8S 4K1, Canada\\
\small \tt soltys@mcmaster.ca
}
\date{\today}

\begin{document}
\maketitle
\begin{abstract}
A shuffle of two strings is formed by interleaving the characters
into a new string, keeping the characters of each string in order. A
string is a {\em square} if it is a shuffle of two identical strings.
There is a known polynomial time dynamic programming algorithm
to determine if a
given string $z$ is the shuffle of two given strings $x,y$; however,
it has been an open question whether there is a polynomial time algorithm
to determine if a given string~$z$ is a square.  We resolve
this by proving that this problem is NP-complete
via a many-one reduction
from 3-{\sc Partition}.
\end{abstract}

\section{Introduction}\label{sec:intro}

If $u$, $v$, and~$w$ are strings over an alphabet~$\Sigma$, then
$w$ is a {\em shuffle} of $u$ and~$v$ provided there are (possibly
empty) strings $x_i$ and $y_i$ such that $u = x_1x_2\cdots x_k$
and $v = y_1y_2\cdots y_k$ and $w = x_1y_1x_2y_2\cdots x_ky_k$.
A shuffle is sometimes instead called a ``merge'' or an ``interleaving''.
The intuition for the definition is that $w$ can be obtained from
$u$ and~$v$ by an operation similar to shuffling two decks of
cards.
We use $w=u\odot v$ to denote that $w$~is a shuffle
of $u$ and~$v$; note, however, that in spite of the
notation there can be many
different shuffles~$w$ of $u$ and~$v$.  The string~$w$ is
called a {\em square} provided it is equal to a shuffle
of a string~$u$ with itself, namely provided $w= u\odot u$ for
some string~$u$.  This paper proves that the set of squares is
NP-complete; this is true even for (sufficiently large) finite
alphabets.

The initial work on shuffles arose out of abstract
formal languages, and shuffles were motivated later by applications
to modeling sequential
execution of concurrent processes.  To the best of
our knowledge, the shuffle
operation was first used in formal languages by
Ginsburg and Spanier~\cite{GinsburgSpanier:TwoTapeDevices}.
Early research with applications to concurrent processes
can be found in
Riddle~\cite{Riddle:1973,Riddle:1978} and
Shaw~\cite{Shaw:FlowExpressions}.
Subsequently,
a number of authors,
including
\cite{gischer-1981,%
GruberHolzer:regular,%
jantzen-1980,%
jantzen-1984,%
Jedrzejowicz:StructuralShuffle,%
JedrzejowiczSzepietowski:ShuffleInP,%
JedrzejowiczSzepietowski:ShuffleRegular,%
MayerStockmeyer:Interleaving,%
ORR:concurrency,%
Shoudai:1992}
have studied various aspects of
the complexity of the shuffle and iterated shuffle
operations in conjunction with regular expression
operations and other constructions from
the theory of programming languages.

In the early 1980's,
Mansfield~\cite{Mansfield:MergeAlgorithm,Mansfield:MergeComplexity}
and Warmuth and Haussler~\cite{WarmuthHaussler:IteratedShuffle}
studied the computational complexity of the shuffle operator
on its own.
The paper~\cite{Mansfield:MergeAlgorithm} gave a polynomial
time dynamic programming
algorithm for deciding the following shuffle
problem: Given inputs $u, v, w$, can $w$ be expressed as
a shuffle of $u$ and~$v$, that is, does $w = u\odot v$?
In~\cite{Mansfield:MergeComplexity}, this was extended to give
polynomial time algorithms for deciding whether a string~$w$ can
be written as the shuffle of $k$~strings $u_1,\ldots, u_k$,
so that $w = u_1 \odot u_2 \odot \cdots \odot u_k$,
for a {\em constant} integer~$k$.  The paper~\cite{Mansfield:MergeComplexity}
further proved that if $k$ is allowed to vary, then the
problem becomes NP-complete (via a reduction from
{\sc Exact Cover with 3-Sets}).
Warmuth and Haussler~\cite{WarmuthHaussler:IteratedShuffle} gave
an independent proof of this last result and went on to
give a rather striking improvement by showing that this problem
remains NP-complete even if the $k$ strings $u_1,\ldots,u_k$
are equal.  That is to say, the question of,
given strings $u$ and~$w$, whether $w$ is equal to
an {\em iterated shuffle} $u\odot u \odot \cdots \odot u$
of~$u$ is NP-complete.
Their proof used a reduction from {\sc 3-Partition}.

The second author~\cite{Soltys:Shuffle} has recently proved that the
problem of
whether $w = u\odot v$ is in $\AC^1$, but not in~$\AC^0$.
Recall that $\AC^0$ (resp., $\AC^1$) is the class of problems recognizable
with constant-depth (resp., logarithmic depth) Boolean circuits.

As mentioned above, a string~$w$ is defined
to be a {\em square} if it can be written $w = u\odot u$
for some~$u$.
Erickson~\cite{Erickson:UnshufflingStackExchange} in 2010,
asked on the {\em Stack Exchange} discussion board
about the computational complexity
of recognizing squares, and in particular
whether this is polynomial
time decidable.   This problem was repeated
as an open question in~\cite{HRS:ShufflingUnshuffling}.
An online
reply to~\cite{Erickson:UnshufflingStackExchange}
by Per Austrin showed that
the problem of recognizing squares
is polynomial time decidable provided that
each alphabet symbol occurs at most four times in~$w$ (by a
reduction from {\sc 2-Sat}); however, the general question has
remained open.  The present paper
resolves this by proving that
the problem of recognizing squares is NP-complete, even over a
sufficiently large fixed alphabet.

The NP-completeness proof uses
a many-one reduction from
the strongly NP-complete problem
{\sc 3-Partition} (see~\cite{GareyJohnson:NPcompleteness}).
{\sc 3-Partition} is
defined as follows: The input is a sequence of natural numbers
$S=\langle n_i:1\le i\le 3m\rangle$ such that
$B=(\sum_{i=1}^{3m}n_i)/m$ is an integer and
$B/4<n_i<B/2$ for each $i\in[3m]$.
The question is: can $S$ be partitioned into $m$
disjoint subsequences $S_1,\ldots,S_m$ such that each
$S_k$ has exactly three elements with the sum
of the three members of~$S_k$ equal to~$B$?
Since {\sc 3-Partition} is {\em strongly} NP-complete, it remains
NP-complete even if the integers~$n_i$ are presented in
unary notation.

\section{Mathematical preliminaries}\label{sec:prelim}

Let $w$ be a string of symbols over the alphabet~$\Sigma$ with
$w = w_1\cdots w_n$ for $w_i\in\Sigma$, so $n = |w|$.
A string~$u$ is a {\em subword} of~$w$ if $w = v_1 u v_2$
for some strings $v_1,v_2$.  A string~$u^\prime$ is
a {\em subsequence} of~$w$ if $w = u^\prime\odot v$ for some string~$v$.
Both the subword~$u$ and the subsequence~$u^\prime$ contain symbols
selected in increasing order from~$w$; the symbols of~$u$
must appear consecutively in~$w$ but this is not required
for~$u^\prime$.  The exponential notation $u^i$, for $i\ge 0$,
indicates the word obtained by concatenating $i$ copies of~$u$.
If $u_1,\ldots,u_k$ are strings, the product notation
$\prod_{\ell=1}^k u_\ell$ indicates the concatenation
$u_1 u_2 \cdots u_{k-1} u_k$.

Now suppose that $w$ is a square.
Figure~\ref{fig:matchingexample} gives an example
of how a square shuffle $w = u \odot u$ gives rise
to a bipartite graph~$G$ on the symbols of~$w$.
The graph~$G$ is defined based on
a particular computation of $w$ as a
shuffle $u\odot u$ as obtained by shuffling
two copies of~$u$.\footnote{In general, there may be
several such ways to express $w$ as a square shuffle,
even for the same~$u$.}
The vertices of~$G$ are the symbols $w_1,\ldots,w_n$
of~$w$, and, for each~$i$, $G$ contains an edge joining
the symbol of~$w$
corresponding to the $i$-th
symbol of one copy of~$u$ to the
symbol of~$w$ corresponding
to the $i$-th symbol of the other copy of~$u$.
W.l.o.g., if $G$ contains an edge joining
$w_j$ and~$w_k$ with $j<k$, then $w_j$~corresponds
to a symbol in the first copy of~$u$,
and $w_k$~corresponds to a symbol in the
second copy of~$u$.  This can be done without loss
of generality, possibly by changing the order in
which the symbols of the $u$'s are shuffled out
to form~$w$.  (So we could instead define $G$ as
a {\em directed} graph if we wished.)

\begin{figure}[t]
\begin{center}
\psset{unit=1.3pt}
\begin{pspicture}(6,0)(230,20)
\rput(0,0){$c_1$}
\rput(10,0){$x\vphantom{c_1}$}
\rput(20,0){$x\vphantom{c_1}$}
\rput(30,0){$x\vphantom{c_1}$}
\rput(40,0){$c_2$}
\rput(50,0){$c_1$}
\rput(60,0){$x\vphantom{c_1}$}
\rput(70,0){$x\vphantom{c_1}$}
\rput(80,0){$x\vphantom{c_1}$}
\rput(90,0){$c_2$}
\rput(100,0){$c_1$}
\rput(110,0){$x\vphantom{c_1}$}
\rput(120,0){$x\vphantom{c_1}$}
\rput(130,0){$c_2$}
\rput(140,0){$c_1$}
\rput(150,0){$x\vphantom{c_1}$}
\rput(160,0){$x\vphantom{c_1}$}
\rput(170,0){$c_2$}
\rput(180,0){$c_1$}
\rput(190,0){$x\vphantom{c_1}$}
\rput(200,0){$c_2$}
\rput(210,0){$c_1$}
\rput(220,0){$x\vphantom{c_1}$}
\rput(230,0){$c_2$}
\psbezier(0,5)(0,15)(50,15)(50,5)  
\psbezier(40,5)(40,35)(130,35)(130,5)  
\psbezier(90,5)(90,35)(170,35)(170,5) 
\psbezier(100,5)(100,35)(180,35)(180,5) 
\psbezier(140,5)(140,35)(210,35)(210,5) 
\psbezier(200,5)(200,15)(230,15)(230,5) 
\psbezier(10,5)(10,25)(60,25)(60,5)
\psbezier(20,5)(20,25)(70,25)(70,5)
\psbezier(30,5)(30,20)(110,20)(110,5)
\psbezier(80,5)(80,20)(150,20)(150,5)
\psbezier(120,5)(120,20)(190,20)(190,5)
\psbezier(160,5)(160,20)(220,20)(220,5)
\end{pspicture}
\end{center}
\caption{Let $w$ be the string $(c_1x^3c_2)^2(c_1 x^2 c_2)^2 (c_1 x c_2)^2$.
This figure shows the bipartite graph~$G$ associated
with the square shuffle $w = u\odot u$ with $u$ equal
to $c_1 x x x c_2 x c_2 c_1 x c_1 x c_2 $.
It is not pictured, but we also have $w = v \odot v$
with $v = c_1x^3c_2c_1 x^2 c_2c_1 x c_2$.
\newline
\hspace*{1.5em}The string~$w$ can be expressed
in product notation as
$\prod_{k=0}^2(c_1x^{3-k}c_2)^2$.}
\label{fig:matchingexample}
\end{figure}
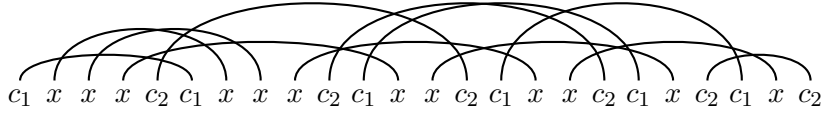

The bipartite graph~$G$ has a special ``non-nesting'' property:
if $G$ contains an edge from $w_k$ to~$w_\ell$ and an edge
from $w_p$ to~$w_q$, then it is not the case that
$k<p<q<\ell$.  This is because there are
indices $i$ and~$i^\prime$ such that
$w_k$ and~$w_\ell$
correspond to the $i$-th symbols of the first and second
copies of~$u$, and such that $w_p$ and~$w_q$ correspond to the
$i^\prime$-th symbols of the two copies of~$u$. (Compare to
Figure~\ref{fig:nesting}.)  But then
$k<p$ implies $i<i^\prime$ whereas $q<\ell$ implies that
$i^\prime<i$, and this is a contradiction.

\begin{figure}[t]
\begin{center}
\psset{unit=1.3pt}
\begin{pspicture}(0,6)(60,20)
\rput(0,0){$a$}
\rput(20,0){$b$}
\rput(40,0){$a$}
\rput(60,0){$b$}
\psbezier(0,5)(0,15)(40,15)(40,5)  
\psbezier(20,5)(20,15)(60,15)(60,5)  
\end{pspicture}
\hspace*{0.7in}
\begin{pspicture}(0,6)(60,20)
\rput(0,0){$a$}
\rput(20,0){$b$}
\rput(40,0){$b$}
\rput(60,0){$a$}
\psbezier(0,5)(0,20)(60,20)(60,5)  
\psbezier(20,5)(20,10)(40,10)(40,5)  
\end{pspicture}
\end{center}
\caption{Examples of two crossing (and hence non-nested) edges
for a graph on $abab$, and two nested edges for a graph
on $abba$.  Nested edges
cannot appear in a graph obtained from a shuffle.}
\label{fig:nesting}
\end{figure}
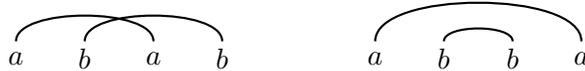

In fact, as is easy to prove, if
there is a complete bipartite graph~$G$
of degree one (i.e., a perfect matching)
on the symbols of~$w$
which is non-nesting, then $w$ can be expressed as a square shuffle
$w = u\odot u$ so that $G$~is the bipartite graph associated
with this shuffle.

The non-nesting property for~$G$ can also be
viewed as an ``anti-Monge'' condition, namely
as the opposite of the Monge condition.
A bipartite graph on the symbols
of the string~$w$ is said to satisfy the {\em Monge condition}
provided that, instead of having the non-nesting condition, it is
prohibited that $k<p<\ell<q$.  In other words, the
Monge condition allows nested edges but prohibits crossing
edges.
The Monge condition has been widely studied
for matching problems and transportation problems.
Many problems that satisfy the Monge condition
or the ``quasi-convex'' condition
are known to have efficient polynomial time
algorithms; for these
see~\cite{BussYianilos:quasiconvex} and the
references cited therein.  There are fewer algorithms
known for problems that satisfy the anti-Monge property,
and some special cases are known to be NP-hard~\cite{BCRW:AntiMonge}.
This is another reason why we find the NP-completeness
of the square problem to be interesting: it provides a
hardness result for anti-Monge matching in a very
simple and abstract situation.

The set of squares~$w$ is accepted by
the following finite-state queue automaton.
A queue automaton is defined similarly to
a PDA but
with a queue instead of a stack.
As usual, the automaton reads the
input~$w$ from left to right.  The automaton's queue
is initially empty and supports the operations
push-right (enqueue) and pop-left (dequeue).
The automaton accepts if its queue is empty
after the last symbol of~$w$ has been read.  The
non-deterministic algorithm for the automaton is as follows:

\vbox{
\begin{tabbing}
\hspace*{1em}\=Repeatedly do one of the following:
\\
\>\hspace*{1em}\=a.~Read the next input symbol~$\sigma$
and push it onto the queue, or \\
\>\>b. \=If the next input symbol~$\sigma$ is the same
as the symbol at the top
of \\
\>\>\>the queue, read past the input symbol~$\sigma$
and pop $\sigma$ from the queue.
\end{tabbing}
}
When either step a.\ or~b.\ is performed, we say that the
input symbol~$\sigma$ has been {\em consumed}.  In case~b.,
we say that the symbol~$\sigma$ on the queue has been
{\em matched} by the input symbol. Note that a.\ is always
allowed, and~b.\ only when the symbols match.

A {\em configuration} of the automaton is a ``snapshot''
of the
computation, and consists of the queue contents~$Q$
and the remaining part~$x$ of the input to be read.
A configuration is denoted $Q\|x$.  A single step
from configuration~$C$ to configuration~$C^\prime$ is denoted
$C\vdash C^\prime$.  A sequence of zero or more steps is
denoted $C\vdashstar C^\prime$.  The condition
$C\vdash C^\prime$ can hold in one of two ways:
if $C$ is $Q \| \sigma x$, then either (a)~$C^\prime$
is $Q\sigma \| x$, or (b)~$C^\prime$ is $Q^\prime \| x$
where $Q = \sigma Q^\prime$.  The input~$w$ is {\em accepted}
if $\varepsilon \| w \,\vdashstar \varepsilon \| \varepsilon$,
where $\varepsilon$ is the empty string.
More generally, a configuration~$C$ is {\em accepted}
provided $C \vdashstar \varepsilon \| \varepsilon$.

If a computation proceeds as
\begin{equation}\label{eq:consume}
u_1 u_2 u_3 \| x_1 x_2 x_3 ~\vdashstar~
    u_2 u_3 z_1 \| x_2 x_3 ~\vdashstar~
    u_3 z_1 z_2 \| x_3,
\end{equation}
then we say that the subword~$x_2$ of the input is
{\em consumed} by the subword~$u_2$ of the queue.  This means
that the symbols of~$x_2$ are either matched
against symbols from~$u_2$, or are pushed onto the
queue only after all the symbols
of~$u_1$ have been popped and before any symbol of~$u_3$ is
popped.  In addition, no symbol of $x_1$ or~$x_3$ is matched
against a symbol from~$u_2$.  The
word~$z_2$ which is pushed onto the stack while $x_2$ is
consumed by~$u_2$ is called the {\em resultant}.\footnote{Note
that in (\ref{eq:consume}) also $x_1$ is consumed by~$u_1$ with
resultant~$z_1$.}
The following two simple lemmas, which will be used in the next
section, illustrate these concepts.
\begin{lem}\label{lem:consumesubsequence}
If $x_2$ is consumed by~$u_2$ yielding the resultant~$z_2$,
then $u_2$ and~$z_2$ are subsequences of~$x_2$.  Furthermore,
$x_2 = u_2 \odot z_2$.
\end{lem}
\begin{proof}
This holds since $u_2$ is equal to the subsequence
of symbols of~$x_2$ that are matched against symbols
of~$u_2$, and $z_2$~is
the subsequence of symbols of~$x_2$ which are enqueued
and so not matched against
symbols from~$u_2$.
\end{proof}
\begin{lem}\label{lm:e0e}
Suppose $e_0,e$ are symbols that do not appear
in the strings $u_i$, $x_i$, or~$v$.
Consider the string
$w = e_0 u_1 e u_2 e \cdots e u_k e e_0 x_1 e x_2 e \cdots e x_k e v$.
Any accepting computation of~$w$ must proceed as:
\begin{eqnarray}
\label{eq:e0e}
\varepsilon \| w
   & \vdashstar
         & u_1 e u_2 e u_3 e \cdots  e u_k e \| x_1 e x_2 e x_3 e \cdots e x_k e v \\
   \nonumber
   & \vdashstar & u_2 e u_3 e \cdots e u_k e z_1 \| x_2 e x_3 e \cdots e x_k e v \\
   \nonumber
   & \vdashstar & u_3 e \cdots e u_k e z_1 z_2 \| x_3 e \cdots e x_k e v \\
   \nonumber
   & \vdashstar & u_k e z_1\cdots z_{k-1} \| x_k e v
   ~\vdashstar~ z_1\cdots z_k \| v
   ~\vdashstar \varepsilon \| \varepsilon,
\end{eqnarray}
so that each $x_i$ is consumed by the corresponding~$u_i$ with resultant~$z_i$.
\end{lem}
\begin{proof}
The two occurrences of $e_0$ must be matched with each other during
the accepting computation.  By the non-nesting property, this means
that all the symbols between the two $e_0$'s must be pushed onto the
queue instead of matching any prior symbol.  At this point, there
are exactly $k$~many $e$'s on the queue and an equal number of $e$'s
remaining in the input.  The non-nesting
property thus implies that the $i$-th occurrence of~$e$
pushed onto the queue must be matched against the
$i$-th occurrence of~$e$ in the second half of~$w$.  From this
it is evident, again by the non-nesting property,
that the accepting computation follows the
pattern~(\ref{eq:e0e}); therefore each $x_i$ is consumed by~$u_i$.
\end{proof}

\section{Main Result}

\begin{thm}\label{thm:squareNPcomplete}
The set {\sc Square} of squares is NP-complete.  This is true even for
sufficiently large finite alphabets.
\end{thm}

We shall prove the theorem for an alphabet with
9 symbols.  A relatively straightforward modification of our
proof shows that the theorem also holds for alphabets of size~7.
We conjecture that Theorem~\ref{thm:squareNPcomplete}
holds even for alphabets of size~2, but this would require
substantially new proof techniques.  (Over a unary alphabet,
{\sc Square} is just the set of even length strings.)

The rest of the paper is devoted to the proof of
Theorem~\ref{thm:squareNPcomplete}.
Clearly the set of squares is in NP.  To prove
the NP-completeness, we shall give a
logspace computable many-one reduction
from {\sc 3-Partition} to {\sc Square}.

Consider an instance of {\sc 3-Partition}
$S=\langle n_i:1\le i\le
3m\rangle$ such that the $n_i$'s are given
in unary notation and
such that $B=(\sum_{i=1}^{3m}n_i)/m$ is an integer.
We also have $B/4<n_i<B/2$ for each~$i$, but shall not
use this fact.  Without loss of generality,
the values~$n_i$ are given in {\em non-increasing} order
(if not, then reorder them).
The many-one reduction to {\sc Square}
constructs a string $w_S$ over the
alphabet
\[
\Sigma=\{a_1,a_2,b,e_0,e,c_1,c_2,x,y\},
\]
such that $w_S$ is a
square iff $S$~is a ``yes'' instance of {\sc 3-Partition}.
The string~$w_S$ consists of three parts:
\[
w_S ~ := ~ \load{S}\dist{S}\ver{S}.
\]
These are defined by
\begin{align*}
\load{S} &~=~ e_0\prod_{i=1}^m(b^{2B}e) \\
\dist{S} &~=~ e_0\prod_{i=1}^m((a_1b^Ba_2)^3e) \\
\ver{S} &~=~
   \prod_{k=1}^{3m} \left[
         v_{4k-3}D_k v_{4k-3}
         v_{4k-2}D_k v_{4k-2}
         v_{4k-1}E_k v_{4k-1}
         v_{4k}  F_k v_{4k}
   \right]
\end{align*}
where
\begin{eqnarray*}
v_\ell &=& c_1 x^\ell y^\ell c_2 \\
D_k &=& (a_1^2 b^{n_k} a_2^2)^{3m-k+1} \\
E_k &=& (a_1^2 b^B a_2^2)^{3m-k} (a_1 b^{n_k} a_2) (a_1^2 b^B a_2^2)^{3m-k} \\
F_k &=& (a_1^2 b^B a_2^2)^{2(3m-k)} \\
\end{eqnarray*}
It is useful to let $U_\ell := a_1^2 b^\ell a_2^2$ as this lets
us shorten the expressions for $D_k$, $E_k$, and~$F_k$, so
$D_k = U_{n_k}^{3m-k+1}$,
$E_k = U_B^{3m-k}a_1 b^{n_k} a_2 U_B^{3m-k}$,
and $F_k = U_B^{2(3m-k)}$.

The length of $w_S$ is quadratic in $m + \sum_i n_i$,
so $w_S$ is polynomially bounded.  It is clear that $w_S$ can be
constructed from~$S$ by a logspace computation.

The actions of the loader and distributor are relatively
easy to understand, so we describe them first.
As the next lemma states,
the intended function of the loader is to place $m$ many blocks of
$2B$~many $b$'s, separated by~$e$'s, onto the queue.
\begin{lem}\label{lem:loader}
Any accepting computation for $w_S$ starts off
as
\[
\varepsilon \| w_S ~\vdashstar~ e_0 (b^{2B} e)^m \| \dist{S} \ver{S}.
\]
In the subsequent part of
the accepting computation, the $i$-th occurrence of
the subword $(a_1 b^B a_2)^3$ in $\dist{S}$ will be consumed
by the $i$-th occurrence of~$b^{2B}$ in the
queue.
\end{lem}
\begin{proof}
This is an immediate consequence of Lemma~\ref{lm:e0e}
since there are only two occurrences of~$e_0$ in~$w_S$,
and since $e$ has the same number of occurrences between the
two~$e_0$'s as after the second~$e_0$.
\end{proof}

Consider how the subword $(a_1 b^B a_2)^3$ can
be consumed by $b^{2B}$.  Since there are no
$a_1$'s or~$a_2$'s in $b^{2B}$, the $a_1$'s and~$a_2$'s
must be pushed onto the queue.  In addition, exactly
$2B$ of the $3B$ many occurrences of~$b$
in $(a_1 b^B a_2)^3$ must be matched
against the symbols of $b^{2B}$.  Thus, when
the subword $(a_1 b^B a_2)^3$ is consumed by $b^{2B}$
a resultant string of the form
$a_1 b^{j_1} a_2 a_1 b^{j_2} a_2 a_1 b^{j_3} a_2$ must
be pushed onto the queue where $j_1+j_2+j_3 = 3B-2B = B$.
Since the automaton is non-deterministic,
any such values for $j_1,j_2,j_3$ can be achieved.
These observations, together
with Lemma~\ref{lem:loader}, prove
Lemma~\ref{lem:distributor}:

\begin{lem}\label{lem:distributor}
Given any sequence of non-negative
integers $\langle i_k\rangle_{k=1}^{3m}$ such that
\begin{equation}\label{eq:condontrip}
\forall j\in\{1,2,\ldots,m\}, \quad
i_{3j-2}+i_{3j-1}+i_{3j}=B,
\end{equation}
there exists a computation
$\varepsilon\|w_S\,\vdashstar\prod_{k=1}^{3m}(a_1b^{i_k}a_2)\|\ver{S}$.
Conversely, if
$\varepsilon\|w_S\,\vdashstar W\|\ver{S}$
then $W$ must be of the form $\prod_{k=1}^{3m}(a_1b^{i_k}a_2)$, so that
condition~(\ref{eq:condontrip}) holds.
\end{lem}

We now turn to analyzing the effect of $\ver{S}$.
By Lemma~\ref{lem:distributor}, any accepting computation
for $\varepsilon \| w_S$ reaches a configuration
$\prod_{k=1}^{3m}(a_1b^{i_k}a_2)\|\ver{S}$
satisfying (\ref{eq:condontrip}).  The intuition
is that the sets $S_j := \{i_{3j-2},i_{3j-1},i_{3j}\}$
should be a solution to the {\sc 3-Partition} problem~$S$.
By~(\ref{eq:condontrip}), the members of each~$S_j$
sum to~$B$.  Thus, the sets $S_j$ are a solution
to the {\sc 3-Partition} iff the sequence
$\langle i_k \rangle_{k=1}^{3m}$ is a permutation (a reordering)
of~$S = \langle n_k \rangle_{k=1}^{3m}$.

By Lemma~\ref{lem:distributor},
to complete the
proof of Theorem~\ref{thm:squareNPcomplete}, it suffices to
show that $\prod_{k=1}^{3m}(a_1b^{i_k}a_2)\|\ver{S}$
is accepted if and only if the sequence
$\langle i_k \rangle_{k=1}^{3m}$ is a permutation
of~$S$.
We first prove the easier direction
of this equivalence:
\begin{lem}\label{lem:MainIfThen}
Suppose $\langle i_k \rangle_{k=1}^{3m}$ is a permutation
of~$S$.  Then the configuration
$\prod_{k=1}^{3m}(a_1b^{i_k}a_2)\|\ver{S}$
is accepted.  Therefore, if $S$ is a ``Yes'' instance
of {\sc 3-Partition},
then $\varepsilon \,\| w_S \vdashstar \varepsilon \| \varepsilon$
and $w_S$ is in {\sc Square}.
\end{lem}
We prove Lemma~\ref{lem:MainIfThen} after first proving
Lemmas \ref{lem:DkDk} and~\ref{lem:EkFk}.
\begin{dfn}
A computation accepting $w_S$ satisfies the
{\em V-Condition} provided that for each~$\ell$
(for $1\le \ell \le 12m$)
the second occurrence of the
subword $v_\ell$ in~$w_S$ is consumed by
the first occurrence of~$v_\ell$
in~$w_S$.  This means that the symbols of the second~$v_\ell$
are completely matched by those of the first~$v_\ell$.
\end{dfn}
Theorem~\ref{thm:Vcondition} below will prove that the V-Condition
must hold, but for now it suffices to just assume it.

\begin{dfn}
A string~$z$ {\em has $k$~alternations of the symbols
$a_1,a_2$} provided $(a_1a_2)^k$ is a subsequence of~$z$
but $(a_1 a_2)^{k+1}$ is not.
\end{dfn}

\begin{lem}\label{lem:DkDk}
Let $i_1,\ldots,i_{3m-k+1}$ be natural numbers, and
$W = \prod_{j=1}^{3m-k+1} (a_1 b^{i_j} a_2)$.
Suppose the V-Condition holds for a computation containing
the subcomputation
\begin{equation}\label{eq:DkDkLemma}
 W \| v_{4k-3}D_k v_{4k-3} v_{4k-2}D_k v_{4k-2} (\cdots)
 ~\vdashstar~ W^\prime \| (\cdots).
\end{equation}
(The ``$(\cdots)$'' denotes the rest of the input string.)
Then $W^\prime = W$, and $i_j\le n_k$ for all~$j$.
Conversely, if each $i_j\le n_k$,
then the subcomputation~{\rm (\ref{eq:DkDkLemma})}
can be carried out.
\end{lem}
Since $W^\prime = W$, the
computation~(\ref{eq:DkDkLemma}) might seem
to achieve nothing, and thus be
pointless; the point,
however, is that it ensures that the values~$i_j$
are $\le n_k$.  This will be useful for
the proof of Lemma~\ref{lem:MainIf}.
\begin{proof}
By the V-Condition, and the non-nesting property,
the computation~(\ref{eq:DkDkLemma}) must have the form
\[
W \| v_{4k-3}D_k v_{4k-3} v_{4k-2}D_k v_{4k-2} (\cdots)
  ~\vdashstar~ W^\pprime \| v_{4k-2}D_k v_{4k-2} (\cdots)
  ~\vdashstar~ W^\prime \| (\cdots),
\]
where $W^\pprime$ is the resultant
when the first $D_k$ is consumed by~$W$, and
$W^\prime$ is similarly the resultant when
the second $D_k$ is consumed by~$W^\pprime$.

$W$ and $D_k$ both have $3m-k+1$
alternations of $a_1,a_2$.
Therefore, when $D_k$ is consumed by~$W$,
the $j$-th $a_1$ (resp.,~$a_2$) symbol in~$W$
must match an $a_1$ (resp., $a_2$)
from the $j$-th block $a_1 a_1$ (resp, $a_2 a_2$)
in~$D_k$.  The other $a_1$ (resp., $a_2$) in that
block is pushed onto the queue as part of~$W^\pprime$.
Furthermore,
the subword $b^{i_j}$ in the $j$-th component of~$W$
must match $i_j$~of the $b$'s in the $j$-th occurrence
of~$b^{n_k}$ in~$D_k$; this leaves $n_k-i_j$ many $b$'s
to be pushed onto the queue as part of~$W^\pprime$.
This is possible if and only if $i_j \le n_k$ for all~$j$,
and if so,
$W^\pprime = \prod_{j=1}^{3m-k+1} (a_1 b^{n_k-i_j} a_2)$.

The second $D_k$ must be consumed by~$W^\pprime$, and
the same argument shows that this means
$W^\prime = \prod_{j=1}^{3m-k+1} (a_1 b^{i_j} a_2) = W$.
\end{proof}

\begin{lem}\label{lem:EkFk}
Let $i_1,\ldots,i_{3m-k+1}$ be natural numbers, and
$W = \prod_{j=1}^{3m-k+1} (a_1 b^{i_j} a_2)$.
Suppose $i_J = \max_j \{i_j\} = n_k$.
Let $i^\prime_1,\ldots,i^\prime_{3m-k}$
be the sequence $\langle i_j \rangle _j$ with $i_J$ omitted,
and let $W^\prime = \prod_{j=1}^{3m-k} (a_1 b^{i^\prime_j} a_2)$.
Then there is a computation
\begin{equation}\label{eq:EkFkLemma}
 W \| v_{4k-1}E_k v_{4k-1} v_{4k}F_k v_{4k} (\cdots)
 ~\vdashstar~ W^\prime \| (\cdots).
\end{equation}
\end{lem}
The computation~(\ref{eq:EkFkLemma}) will satisfy the
V-Condition.
Lemma~\ref{lem:EkFkconverse} below will prove
a converse to Lemma~\ref{lem:EkFk} under the additional
assumption of the V-Condition.
Lemma~\ref{lem:EkFk}, however, is all that is needed
for Lemma~\ref{lem:MainIfThen}.

\begin{proof}
We construct a computation of the form
\begin{equation}\label{eq:FkFkProof}
W \| v_{4k-1}E_k v_{4k-1} v_{4k}F_k v_{4k} (\cdots)
~\vdashstar~ W^\pprime \| v_{4k}F_k v_{4k} (\cdots)
~\vdashstar~ W^\prime \| (\cdots).
\end{equation}
Recalling that $E_k = U_B^{3m-k}a_1 b^{n_k} a_2 U_B^{3m-k}$
and using $i_J = n_k$,
the first half of the computation~(\ref{eq:FkFkProof}) has the form
\begin{eqnarray*}
\lefteqn{\prod_{j=1}^{3m-k+1} (a_1 b^{i_j} a_2) \| v_{4k-1}U_B^{3m-k}a_1 b^{n_k} a_2 U_B^{3m-k} v_{4k-1} }
\\
 &\vdashstar&
   \prod_{j=J}^{3m-k+1} (a_1 b^{i_j} a_2) v_{4k-1}  \prod_{j=1}^{J-1} (a_1 b^{B-i_j} a_2)
   \| U_B^{3m-k-(J-1)} a_1 b^{n_k} a_2 U_B^{3m-k} v_{4k-1}
\\
 &\vdashstar&
   \prod_{j=J}^{3m-k+1} (a_1 b^{i_j} a_2) v_{4k-1}  \prod_{j=1}^{J-1} (a_1 b^{B-i_j} a_2)
   U_B^{3m-k-(J-1)} \| a_1 b^{n_k} a_2 U_B^{3m-k} v_{4k-1}
\\
 &\vdashstar&
   \prod_{j=J+1}^{3m-k+1} (a_1 b^{i_j} a_2) v_{4k-1}  \prod_{j=1}^{J-1} (a_1 b^{B-i_j} a_2)
   U_B^{3m-k-(J-1)} \| U_B^{3m-k} v_{4k-1}
\\
 &\vdashstar&
   \prod_{j=J+1}^{3m-k+1} (a_1 b^{i_j} a_2) v_{4k-1}  \prod_{j=1}^{J-1} (a_1 b^{B-i_j} a_2)
   U_B^{3m-k} \| U_B^{3m-k-(J-1)} v_{4k-1}
\\
 &\vdashstar&
   v_{4k-1}  \prod_{j=1}^{J-1} (a_1 b^{B-i_j} a_2)
   U_B^{3m-k} \prod_{j=J+1}^{3m-k+1} (a_1 b^{B-i_j} a_2) \| v_{4k-1}
\\
 &\vdashstar&
   \prod_{j=1}^{J-1} (a_1 b^{B-i_j} a_2)
   U_B^{3m-k} \prod_{j=J+1}^{3m-k+1} (a_1 b^{B-i_j} a_2) \| \varepsilon
\\
 &=& \prod_{j=1}^{J-1} (a_1 b^{B-i^\prime_j} a_2) U_B^{3m-k}
     \prod_{j=J}^{3m-k} (a_1 b^{B-i^\prime_j} a_2) \| \varepsilon
 ~=~ W^\pprime \| \varepsilon.
\end{eqnarray*}
The first and fifth steps shown above use the fact that when
$a_1^2 b^B a_2^2$ is consumed by $a_1 b^{i_j} a_2$ the resultant
is $a_1 b^{B-i_j}a_2$ as shown in the proof of Lemma~\ref{lem:DkDk}.
The third step matches $a_1 b^{i_j} a_2$
with the equal $a_1 b^{n_k} a_2$.  The final step matches
$v_{4k-1}$.  The other steps push words $v_{4k-1}$ and $U_B$
from the input to the queue.

The second half of the computation~(\ref{eq:FkFkProof})
proceeds as follows:
\begin{eqnarray*}
\lefteqn{ \prod_{j=1}^{J-1} (a_1 b^{B-i^\prime_j} a_2) U_B^{3m-k}
     \prod_{j=J}^{3m-k} (a_1 b^{B-i^\prime_j} a_2)
     \| v_{4k} (a_1^2 b^B a_2^2)^{2(3m-k)} v_{4k} }
\\
   &\vdashstar& U_B^{3m-k}
     \prod_{j=J}^{3m-k} (a_1 b^{B-i^\prime_j} a_2) v_{4k}
     \prod_{j=1}^{J-1} (a_1 b^{i^\prime_j} a_2)
     \| (a_1^2 b^B a_2^2)^{2(3m-k)-(J-1)} v_{4k} \\
\\
   &\vdashstar&
     \prod_{j=J}^{3m-k} (a_1 b^{B-i^\prime_j} a_2) v_{4k}
     \prod_{j=1}^{J-1} (a_1 b^{i^\prime_j} a_2)
     \| (a_1^2 b^B a_2^2)^{3m-k-(J-1)} v_{4k} \\
\\
   &\vdashstar&
     v_{4k}
     \prod_{j=1}^{3m-k} (a_1 b^{i^\prime_j} a_2)
     \| v_{4k}
   ~\vdashstar~
     \prod_{j=1}^{3m-k} (a_1 b^{i^\prime_j} a_2)
     \| \varepsilon.
\end{eqnarray*}
This is easily seen to be a correct computation. This
proves Lemma~\ref{lem:EkFk}.
\end{proof}

We can now prove Lemma~\ref{lem:MainIfThen}.
Suppose that $S = \langle n_j\rangle_{j=1}^{3m}$
and that $\langle i_j\rangle_{j=1}^{3m}$ is a permutation
of~$\langle n_j\rangle_{j=1}^{3m}$ witnessing that
$S$ is a ``Yes'' instance of {\sc 3-Partition}.
Let $W_k$~be the string
$\prod_{j=1}^{3m-k+1} (a_1 b^{i^\prime_j} a_2)$
where $i^\prime_1,\ldots,i^\prime_{3m-k+1}$ is the
sequence obtained by removing $k-1$ of the largest
elements of the sequence $\langle i_j\rangle_{j=1}^{3m}$.
(When there are multiple equal values $i_j$, they
can be removed from the sequence in arbitrary fixed
order, say
according to the order they appear in the sequence).
The $n_k$'s are non-increasing, so
the maximum~$i^\prime_j$ is equal to $n_k$.
Therefore, Lemmas \ref{lem:DkDk} and~\ref{lem:EkFk} imply
that
\[
W_k \| v_{4k-3}D_k v_{4k-3}
         v_{4k-2}D_k v_{4k-2}
         v_{4k-1}E_k v_{4k-1}
         v_{4k}  F_k v_{4k}
    ~\vdashstar~ W_{k+1} \| \varepsilon.
\]
Combining these computations for $1\le k \le 3m$ gives
$W_1\| \ver{S} \,\vdashstar \varepsilon\|\varepsilon$.
Lemma~\ref{lem:distributor} gives
$\varepsilon \| w_S \vdashstar  W_1\|\ver{S} $.
Thus $\varepsilon\|w_S \,\vdashstar \varepsilon\|\varepsilon$.
This proves Lemma~\ref{lem:MainIfThen}. \hfill $\qed$
\bigskip

The next lemma gives the converse of Lemma~\ref{lem:MainIfThen},
under the assumption that the V-Condition holds.  This, together
with Theorem~\ref{thm:Vcondition} stating that the V-Condition
must hold,
will prove Theorem~\ref{thm:squareNPcomplete}.
\begin{lem}\label{lem:MainIf}
Let $S$ be an instance of {\sc 3-Partition}
and $\langle i_k \rangle_{k=1}^{3m}$ satisfy the conditions
of Lemma~\ref{lem:distributor}
and $W = \prod_{k=1}^{3m}(a_1b^{i_k}a_2)$.
Suppose that
$W\|\ver{S} \,\vdashstar \varepsilon \| \varepsilon$
with a computation that satisfies the V-Condition,
so
$\varepsilon \,\| w_S \vdashstar \varepsilon \| \varepsilon$
and $w_S$ is in {\sc Square}.
Then $S$~is a ``Yes'' instance
of {\sc 3-Partition}.
\end{lem}
The main new tool needed for proving Lemma~\ref{lem:MainIf}
is a converse of Lemma~\ref{lem:EkFk}:
\begin{lem}\label{lem:EkFkconverse}
Let $1\le k\le 3m$,
let $i_1,\ldots,i_{3m-k+1}$ be natural numbers, and
$W_k = \prod_{j=1}^{3m-k+1} (a_1 b^{i_j} a_2)$.
Suppose that $\max_j\{i_j\} \le n_k$.
Further suppose there is a computation
\begin{equation}\label{eq:EkFkLemmaConverse}
W_k \| v_{4k-1}E_k v_{4k-1} v_{4k}F_k v_{4k} (\cdots)
~\vdashstar~ W_{k+1} \| (\cdots)
\end{equation}
that satisfies the V-Condition.
Then there is a $J$ such that $i_J = \max_j \{i_j\} = n_k$
such that, letting
$i^\prime_1,\ldots,i^\prime_{3m-k}$
be the sequence $\langle i_j \rangle _j$ with $i_J$ omitted,
we have $W_{k+1} = \prod_{j=1}^{3m-k} (a_1 b^{i^\prime_j} a_2)$.
\end{lem}

Before we prove Lemma~\ref{lem:EkFkconverse},
we indicate how it, and the V-Condition assumption,
imply Lemma~\ref{lem:MainIf} and thus
imply Theorem~\ref{thm:squareNPcomplete}.
Suppose $C$ is a computation
$\varepsilon\|w_S \,\vdashstar \varepsilon \| \varepsilon$
that obeys the V-Condition.  For $1\le k \le 3m+1$,
define the strings
$V_k$ to be such that $C$ contains the configurations
\[
V_k ~\|~ \prod_{\ell=k}^{3m} \left[
         v_{4\ell-3}D_\ell v_{4\ell-3}
         v_{4\ell-2}D_\ell v_{4\ell-2}
         v_{4\ell-1}E_\ell v_{4\ell-1}
         v_{4\ell}  F_\ell v_{4\ell} \right].
\]
Of course, these $V_k$'s are the intermediate queue contents
as $\ver{S}$ is consumed.  For $1\le k \le 3m$,
define
$V_k^\prime$ to be the strings such that $C$ contains
the configuration
\begin{eqnarray*}
\lefteqn{V^\prime_k ~\|~ v_{4k-1}E_k v_{4k-1}
         v_{4k}  F_k v_{4k}}
\\
&& \quad\quad\quad
      \prod_{\ell=k+1}^{3m} \left[
         v_{4\ell-3}D_\ell v_{4\ell-3}
         v_{4\ell-2}D_\ell v_{4\ell-2}
         v_{4\ell-1}E_\ell v_{4\ell-1}
         v_{4\ell}  F_\ell v_{4\ell} \right].
\end{eqnarray*}
\begin{clm}\label{clm:MainThmPf}
We have:
\begin{description}
\setlength{\parskip}{0pt}
\setlength{\itemsep}{0pt}
\setlength{\parsep}{0pt}
\item[\rm (a)] $V_1$ is equal to
$\prod_{j=1}^{3m-k+1}(a_1 b^{i_j} a_2)$
for some sequence $\langle i_j \rangle_{j=1}^{3m}$
satisfying~{\rm (\ref{eq:condontrip})}.
\item[\rm (b)] For $1\le k\le 3m+1$,
$V_k$ equals $\prod_{j=1}^{3m-k+1}(a_1 b^{i^\prime_j} a_2)$
for some sequence $\langle i^\prime_j\rangle_j$ which
is obtained from $\langle i_j \rangle_{j=1}^{3m}$
by removing (instances of) the
$k-1$ largest entries of $\langle n_j\rangle_{j=1}^{3m}$.
\item[\rm (c)] For $1\le k \le 3m$,
$V^\prime_k$ equals~$V_k$,
and its maximum $i^\prime_j$ value
is less than or equal to~$n_k$.
\end{description}
\end{clm}
The claim is proved by induction on~$k$.  Part~(a), and the
equivalent $k=1$
case of~(b), follows from Lemma~\ref{lem:distributor}.
Part~(c) for a given~$k$ follows
from Lemma~\ref{lem:DkDk} and
from the induction hypothesis that
(b) holds for
the same value of~$k$.  Part~(b) for $k>1$
follows from Lemma~\ref{lem:EkFkconverse} and
from the induction hypothesis that (b) and~(c) hold
for $k-1$.
Since $V_{3m+1} = \varepsilon$,
part~(b) implies that the sequence
$\langle i_j\rangle_{j=1}^{3m}$
is a reordering of
$\langle n_j\rangle_{j=1}^{3m}$.
And, since (\ref{eq:condontrip}) holds,
$\langle i_j\rangle_{j=1}^{3m}$ witnesses that
$S$ is a ``Yes'' instance of {\sc 3-Partition}.
This completes the proof of Lemma~\ref{lem:MainIf},
and thereby Theorem~\ref{thm:squareNPcomplete},
modulo the proofs of Lemma~\ref{lem:EkFkconverse} and
Theorem~\ref{thm:Vcondition}.
\hfill $\qed$

\begin{proof}
(of Lemma~\ref{lem:EkFkconverse}.)
Consider a particular computation~$C$ as
in~(\ref{eq:EkFkLemmaConverse}) that satisfies
the V-Condition.  $C$ has
the form
\[
W_k \| v_{4k-1}E_k v_{4k-1} v_{4k}F_k v_{4k} (\cdots)
~\vdashstar~ Z \| v_{4k}F_k v_{4k}
~\vdashstar~ W_{k+1} \| (\cdots)
\]
where $Z$ is the resultant of $E_k$ being
subsumed by $W_k$.  By assumption, $W_k$
has $3m-k+1$ alternations of $a_1, a_2$, whereas
$E_k$ has $2(3m-k)+1$ and $F_k$ has $2(3m-k)$.
The string $E_k$ is a concatenation of ``blocks''
of the form $a_1 b^{n_k} a_2$
or the form $U_B = a_1^2 b^B a_2^2$.
Each subword
$a_1 b^{i_j} a_2$ in~$W$ has its symbol~$a_1$
matched by some~$a_1$ in~$E_k$ and its $a_2$ matched
by some $a_2$ in the same block or a later block
of~$E_k$: these symbols $a_1$ and~$a_2$ in~$E_k$
determine a contiguous sequence of blocks in~$E_k$
which is consumed by $a_1 b^{i_j} a_2$.  We call
these blocks the ``$j$-consumed'' portion of~$E_k$,
and denote it $Y_j$.
The resultant of $a_1 b^{i_j} a_2$ and its
$j$-consumed portion is denoted~$Z_j$.
There may also be blocks of~$E_k$ which are not
part of any $j$-consumed portion, and these
are called ``non-matched'' blocks of~$E_k$.
The string $Z$ is then the concatenation
of the words~$Z_j$, for $1\le j\le 3m-k+1$,
interspersed with the
non-matched blocks of~$E_k$.

Let us consider the possible resultants~$Z_j$.
We can write $E_k$ as $E_k = P_1 P_2 P_3$
where $P_1 = P_3 = U_B^{3m-k}$ and
$P_2 = a_1 b^{n_k} a_2$.  There are several cases
to consider.
\begin{description}
\item[{\em Case a.}] $Y_j$
is $(a_1^2 b^B a_2^2)^\ell$ for
some $\ell\ge 1$, and thus is a
subword of either $P_1$ or~$P_3$
in~$E_k$.  When $Y_j$ is consumed
by $a_1 b^{i_j}a_2$, one of the two initial~$a_1$'s,
any $i_j$ of the $b$'s,
and then one of the two final~$a_2$'s are matched;
the remaining symbols of~$Y_j$ become
the resultant~$Z_j$ and are pushed onto the
queue.  Therefore, $Z_j$ is equal to
\begin{equation}\label{eq:caseA}
Z_j ~=~
   a_1 b^{B-m_1}
   \prod_{s=2}^\ell (a_2^2 a_1^2 b^{B-m_s})
   a_2
\end{equation}
where $m_1 + m_2 + \cdots + m_\ell = i_j$.
Note that $Y_j$ and~$Z_j$ both have $\ell$ alternations
of $a_1, a_2$.
\item[{\em Case b.}] $Y_j$ spans from $P_1$ to $P_3$
and equals
$(a_1^2 b^B a_2^2)^{\ell_1} a_1 b^{n_k} a_2 (a_1^2 b^B a_2^2)^{\ell_2}$
where $\ell_1,\ell_2\ge 1$.  Arguing as in the previous case,
$Z_j$ is equal to
\[
a_1 b^{B-m_1}
   \prod_{s=2}^{\ell_1} (a_2^2 a_1^2 b^{B-m_s})
   a_1 b^{n_k-m_{\ell_1+1}} a_2
   \prod_{s=\ell_1+2}^{\ell_1+\ell_2+1} (a_2^2 a_1^2 b^{B-m_s})
   a_2
\]
where $m_1 + m_2 + \cdots + m_{\ell_1+\ell_2+1} = i_j$.
In this case, $Y_j$ and $Z_j$ both have
$\ell_1+\ell_2$ alternations of $a_1, a_2$.
\item[{\em Case c.}] $Y_j$ is $a_1 b^{n_k} a_2$,
namely, $Y_j = P_2$.
In this case, $Z_j$ is equal to just $b^{n_k-i_j}$.
If $i_j = n_k$, then $Z_j$ is just $\varepsilon$: this
is called a ``full cancellation'' case.
Note that $Z_j$ has zero alternations of $a_1,a_2$,
whereas $Y_j$ has one alternation.
\item[{\em Case d.}] $Y_j$ is
$(a_1^2 b^B a_2^2)^\ell a_1 b^{n_k} a_2$.
We now have
\begin{equation}\label{eq:caseD}
Z_j ~=~ a_1 b^{B-m_1} \prod_{s=2}^\ell (a_2^2a_1^2 b^{B-m_s})
        a_2^2 a_1 b^{n_k-m_{\ell+1}}.
\end{equation}
where $m_1+\cdots+m_{\ell+1} = i_j$.
$Z_j$ consists of a part with $\ell$ alternations
of~$a_1,a_2$ followed by a subsequent $a_1$ (and possibly $b$'s).
In the ``full cancellation'' case, $m_{\ell+1} = n_k$, and since $i_j \le n_k$,
we have $n_k = m_{s+1} = i_j$ and, for $s\le \ell$, $m_s = 0$.
Otherwise, $Z_j$ ends with one or more $b$'s.
\item[{\em Case e.}] The case where $Y_j$ is
$a_1 b^{n_k} a_2 (a_1^2 b^B a_2^2)^\ell$
is completely analogous to case~d., and we omit it.
\end{description}
For simplicity, let's assume for the moment that
neither case d.\ nor~e.\ occurs.
This means that there is at most one occurrence
of either case b.\ or~c., and the rest of the
cases are case~a.
In cases a.\ and~b., $Z_j$ has the same number
of alternations of $a_1, a_2$ as $Y_j$.
Of course the number of alternations in the non-matched
blocks does not change.  Therefore, $Z$ has
$2(3m-k)+1$ alternations of
$a_1,a_2$ if case~c. does not occur,
and has $2(3m-k)$ alternations if
case~c.\ does occur.
The word~$F_k$ has $2(3m-k)$ alternations
of $a_1,a_2$, and since $F_k$ is consumed
by~$Z$, Lemma~\ref{lem:consumesubsequence}
implies that $Z$ cannot have more alternations
of $a_1, a_2$ than~$F_k$.  Therefore, it must be that
case~c.\ occurs and case~b.\ does not.

We claim that case~c.\ must occur as a
full cancellation case.  If not, then $Z$ will consist
of a subword with $3m-k$ alternations
of $a_1,a_2$ that came from~$P_1$,
followed by some non-zero number of~$b$'s
from the $Z_j$ of case~c., and then by another subword with $3m-k$
alternations of $a_1,a_2$ that came from~$P_3$.
In other words, $(a_1 a_2)^{3m-k} b (a_1 a_2)^{3m-k}$
is a subsequence of~$Z$.
It is not, however, a subsequence of~$F_k$, contradicting
the fact that
$F_k$ is consumed by~$Z$.  If follows
that case~c.\ must have occurred in the full cancellation
version.  Let $J$ be the value of~$j$ for which case~c.\
occurred; since it was a case of full cancellation,
$i_J = n_k$.

Therefore, $Z$ has $2(3m-k)$ alternations of $a_1,a_2$,
and is the concatenation of the $3m-k$ many $Z_j$'s
that arose in case~a.\ (the empty~$Z_J$ has been dropped)
and of zero or more non-matched $a_1^2 b^B a_2^2$'s.
The fact that $F_k$ and $Z$ both have $2(3m-k)$ alternations
of $a_1, a_2$, means that the way $F_k$ can be consumed by~$Z$
is tightly constrained.  First, any non-matched block
$a_1^2 b^B a_2^2$ in~$Z$ must consume (and fully match)
an identical block in~$F_k$ leaving a resultant of~$\varepsilon$.
Second, any $Z_j$ with $\ell$ alternations of $a_1,a_2$ will
be of the form~(\ref{eq:caseA}) and must consume a subword
$G_j = (a_1^2 b^B a_2^2)^\ell$ of~$F_k$.  The first $a_1$ of~$Z_j$
must match one of the two first $a_1$'s of~$G_j$;
the final $a_2$ of~$Z_j$ must match one the final two $a_2$'s
of~$G_j$; the other subwords $a_1^2$ and $a_2^2$ of~$Z_j$ must
match identical subwords in~$G_j$; and the $\ell B-i_j$~many $b$'s in~$Z_j$
all must match $b$'s in~$G_j$.  This can always be done, no matter
what the values of the~$m_s$'s in $Z_j$ are.
Since $G_j$ has $\ell B$~many $b$'s,
the consumption of~$G_j$ by~$Z_j$ yields a
resultant $W^\prime_j$ equal to $a_1 b^{i_j} a_2$.

It follows that, when $F_k$ is consumed by $Z$, the resultant equals
the concatenation of the strings $W^\prime_j = a_1 b^{i_j} a_2$,
omitting the word $w_J$ (which triggered case~c.).  In other words,
the resultant is
just~$W_{k+1}$,
proving Lemma~\ref{lem:EkFkconverse} in this case.

We still have to consider the case where case d.\ or~e. occurs.
The cases are symmetric, so suppose case~d.\ occurs, and thus
the rest of the $Z_j$'s are generated by case~a.
Suppose $Z_j$ is obtained via case~d.,
and so is equal to~(\ref{eq:caseD}).  We claim that
this must be a full cancellation case of case~d., with
$n_k = i_j$.  If not, then $Z$ contains $3m-k$ alternations
of $a_1,a_2$ up through $Z_j$, followed by the final $a_1$
of~$Z_j$ and at least one $b$ at the end of~$Z_j$, and then
followed by $3m-k$ alternations of $a_1,a_2$ in the remaining
part of~$Z$.  In other words, $(a_1a_2)^{3m-k} a_1 b (a_1a_2)^{3m-k}$
is a subsequence of~$Z$.  It is not
a subsequence of~$F_k$ however, contradicting
the fact that $F_k$ is to be
consumed by~$Z$.  Thus we must have a full cancellation case
of case~d.

Now consider
what immediately follows $Z_j$ in~$Z$.  It must
either be of the form $a_1^2 b^B a_2^2$ (obtained
from a non-matched block), or, referring to~(\ref{eq:caseA}),
be the word of the
form
\[
Z_{j+1}
  ~=~ a_1 b^{B-m^\prime_1}
      \prod_{s=2}^{\ell^\prime} (a_2^2 a_1^2 b^{B-m^\prime_s})a_2.
\]
obtained from case~a.\ for $Y_{j+1}$.
We claim it is impossible for $Z_j a_1^2 b^B a_2^2$ to
be a subword of~$Z$.  If so,
$(a_1a_2)^{3m-k} a_1^2 (a_1a_2)^{3m-k}$ is a
subsequence of~$Z$, and thus
$Z$ is not a subsequence of~$F_k$.  As before,
this is a contradiction.

We have eliminated the other possibilities,
so $Z_j Z_{j+1}$ is a subword of~$Z$ and $n_k = m_{\ell+1}$.
Therefore, $m_s = 0$ for all $s\le \ell$,
and we have
\[
Z_j Z_{j+1} ~=~ a_1 b^B (a_2^2 a_1^2 b^B)^\ell
         \prod_{s=1}^{\ell^\prime}
         (a_2^2 a_1^2 b^{B-m^\prime_s})a_2.
\]
Note that $Z_j Z_{j+1}$ contains $\ell+\ell^\prime$
alternations of $a_1,a_2$.  Also note that
the subword $Y_j Y_{j+1}$ contains $\ell+\ell^\prime+1$
many such alternations.  Therefore $Z$ has $3m-k$ alternations
of $a_1,a_2$, namely one fewer than~$E_k$ (as desired).
Similarly to the argument four paragraphs above,
it follows that $Z_j Z_{j+1}$ must consume a subword~$G$
of~$F_k$ of the form $(a_1^2 b^B a_2^2)^{\ell+\ell^\prime}$.
Since \hbox{$m_1^\prime+\cdots +m^\prime_{\ell^\prime} = i_{j+1}$},
$Z_j Z_{j+1}$~has $(\ell +\ell^\prime)B-i_j$ many $b$'s.
Hence the resultant when $G$ is consumed by~$Z_j$ is
equal to $a_1 b^{i_{j+1}} a_2$.  If follows again that
when $F_k$ is consumed by~$Z$ it yields the resultant~$W_{k+1}$
as desired.

This completes the proof of Lemma~\ref{lem:EkFkconverse}.
\end{proof}

\paragraph*{The V-Condition.}
The
proof of Theorem~\ref{thm:squareNPcomplete}
will be finalized once we prove
that the V-Condition must hold:
\begin{thm}\label{thm:Vcondition}
Any accepting computation
$\varepsilon || w_S \,\vdashstar \varepsilon \| \varepsilon$
satisfies the V-condition.
\end{thm}
Let
\begin{equation}\label{eq:V}
V ~=~ \prod_{i=0}^{\ell-1} v_{\ell-i}^2
  ~=~ \prod_{j=\ell,\ldots,2,1} (c_1 x^j y^j c_2)^2,
\end{equation}
i.e.,
$V = v_\ell v_\ell \cdots v_2 v_2 v_1 v_1$.
(The dependence of~$V$ on~$\ell$ is suppressed in the notation.)
The
symbols $c_1, x, y, c_2$ occur only in the subwords $v_\ell$
of~$w_S$, and $V$ is the subsequence of~$w_S$ containing
these symbols, but in reversed order.  (We use the reversed
order since it makes the proof below a little simpler to state.)
Clearly, any expression of~$w_S$
as a square shuffle induces a square shuffle
for~$V$.
Therefore Theorem~\ref{thm:Vcondition} is a consequence of
Theorem~\ref{thm:Vsolution}:
\begin{thm}\label{thm:Vsolution}
Let $\ell \ge 1$.
The only accepting computation
$\varepsilon \| V \,\vdashstar \varepsilon \| \varepsilon$
is the one that matches each $v_k$ in~$V$ with the other~$v_k$
in~$V$.
\end{thm}
As a side remark,
it is interesting to note that Figure~\ref{fig:matchingexample}
illustrates that Theorem~\ref{thm:Vsolution} would not hold
if  the $v_j$'s were instead defined to equal $c_1 x^j c_2$
with the $y$'s omitted.
Theorem~\ref{thm:Vsolution} follows from the next three lemmas.
\begin{dfn}
Each subword $x^j$ or~$y^j$ shown in the definition
of~$V$ in~(\ref{eq:V}) is called an {\em $x$-block} or
a {\em $y$-block}, respectively.  We also refer to them
as {\em full $x$-blocks} or {\em full $y$-blocks} after
they have been pushed onto the queue to emphasize that the
complete subword $x^j$ or $y^j$ has been pushed onto the queue
without any $x$ or~$y$ from the block being matched.
\end{dfn}
\begin{lem}\label{lem:xysameblock}
If $C$ is an accepting computation of~$V$,
then $C$ does not match any~$x$ (resp.,~$y$) with another
symbol from the same $x$-block (resp.~$y$-block).
\end{lem}
\begin{proof}
$V$ contains an even number of~$c_1$'s and
an even number of~$c_2$'s.
Consider some $x$- or $y$-block~$\beta$ in~$C$.  There is
either an odd number of $c_1$'s before (and therefore, after)
$\beta$ in~$V$, or an odd number of~$c_2$'s before (and
after) $\beta$ in~$V$.  If there are, say, odd numbers of~$c_1$'s
then some $c_1$ before~$\beta$ must match some~$c_1$ after~$\beta$
during~$C$.  The non-nesting condition now implies that
no two symbols
in~$\beta$ can be matched.
\end{proof}

\begin{lem}\label{lem:Vexiststwofull}
Suppose $C$ is an accepting computation for~$V$,
and $C$ does not completely match the first subword~$v_\ell$
of~$V$ with the second $v_\ell$ of~$V$
(i.e., at least one
symbol from the second~$v_\ell$
of~$V$ is pushed onto the queue).
Then there is a point in~$C$ where the queue contains
either two full $x$-blocks or two full $y$-blocks.
\end{lem}
\begin{proof}
The proof splits into cases depending on how $C$ starts off.
For the first case, suppose the first $c_1$ of~$V$ does not
match the second~$c_1$ of~$V$.  By the non-nesting condition,
this implies that the subword $x^\ell y^\ell c_2 c_1 x^\ell y^\ell$
is pushed onto the queue.  This puts two full $x$-blocks and two full
$y$-blocks on the queue, so the lemma holds in this case.
So, henceforth assume that the first $c_1$ matches the second~$c_1$.

Now suppose the first $x$-block $x^\ell$ does not completely
match the second~$x^\ell$.  Therefore, some of the $x$'s in the
first $x^\ell$ match symbols from some $x^j$ with $j<\ell$. This
$x$-block~$x^j$ comes {\em after} the first two $y$-blocks
(which equal $y^\ell$), so by the non-nesting condition,
these two $y$-blocks are on the queue by the time the
algorithms consumes the $x$-block~$x^j$.  So the lemma holds
in this case as well.  Assume henceforth that the
first $c_1 x^\ell$ is completely matched
with the second $c_1 x^\ell$ by~$C$.

Finally, suppose that the first subword~$y^\ell c_2$ does not
completely match the second~$y^\ell c_2$.  In this case,
we claim that, after consuming the second~$c_2$,
$C$'s queue will contain $y^m c_2 y^m c_2$.  To see this note
that either the two $y^\ell$'s completely match (so $m=0$) and then
the $c_2$'s are not matched by assumption, or the two $y^\ell$'s
do not completely match (so $m>0$) and then the $c_2$'s must be pushed
to the queue since they cannot be matched while a~$y$
is at the top of the queue.  At any rate, the queue contains
two~$c_2$'s once the second~$c_2$ is consumed.  By the non-nesting
property, the second~$c_2$
on the queue must match the {\em fourth}~$c_2$ of~$V$ or a
later~$c_2$ of~$V$. Therefore, the two
$x$-blocks $x^{\ell-1}$ that come prior to the fourth~$c_2$
are pushed
onto the queue, and the lemma holds again in this case.
\end{proof}

\begin{lem}\label{lem:Vnexttwoblocks}
If $C$ is an accepting computation for~$V$ and
at some point in~$C$ the queue
contains two full $x$-blocks (respectively, contains two full
$y$-blocks), then there is a later point at which the queue contains
two full $y$-blocks (respectively, contains two full $x$-blocks).
\end{lem}
\begin{proof}
Suppose $C$ has two full $x$-blocks $x^m$ and then $x^j$ in the queue.
Note $j\le m$.
Let the computation continue until $x^m$ has been
matched, and then until $x^j$ has been matched.
The symbols of $x^m$ are matched by symbols from
$x$-blocks $x^s$ that have $s<m$ (since the block~$x^j$
was intervening).  Therefore, $x^m$'s symbols must match
$x$'s from at least two distinct $x$-blocks.  Between
these two $x$-blocks there is a $y$-block, and by the non-nesting
condition this $y$-block is pushed onto the queue in its
entirety.  Similarly the $x$-block $x^j$ is matched
against symbols from at least two distinct $x$-blocks,
and again there is a $y$-block between those two $x$-blocks
that is entirely pushed onto the queue.  Therefore, once the
$x^j$ is matched, there are
at least two full $y$-blocks in the queue.

The dual argument works with $x$
and $y$ interchanged.
\end{proof}
We can now prove Theorem~\ref{thm:Vsolution}:
\begin{proof}
The proof is by induction on~$\ell$.  The base case $\ell=1$
is trivial.  Suppose $\ell>1$.  If
an accepting computation~$C$ matches the first
two subwords $c_1 x^\ell y^\ell c_2$ against each other
completely, then the rest of the computation~$C$ is
an accepting computation on the rest of~$V$,
namely $V$ minus these first
two subwords.
By the induction hypothesis, the latter accepting
computation matches each pair of subwords $c_1 x^j y^ j c_2$,
and the theorem holds.
Otherwise, if the first two subwords $c_1 x^\ell y^\ell c_2$
of~$V$ are not completely matched by~$C$,
then Lemma~\ref{lem:Vexiststwofull} states that $C$ contains
some point where its queue contains either two full $x$-blocks
or two full $y$-blocks.  Lemma~\ref{lem:Vnexttwoblocks}
then implies that $C$'s queue must contain two full $x$- or
$y$-blocks infinitely often, which is a contradiction.
\end{proof}

That completes the proof of Theorem~\ref{thm:Vsolution},
and thereby the proof of Theorem~\ref{thm:Vcondition},
giving us the V-Condition that was needed for the proof
of Theorem~\ref{thm:squareNPcomplete}.

\bibliographystyle{siam}
\bibliography{shuffle}

\end{document}